\documentclass{article}
\usepackage{amsmath,amssymb,amsthm}
\usepackage[normalem]{ulem}
\usepackage{graphicx,color,epsfig,epstopdf}
\usepackage{fullpage}
\usepackage{caption}
\usepackage{subcaption}
\newcommand{\DEF}{\sl}
\newcommand{\STAB}{\mathop{\mathrm{STAB}}}
\newcommand{\SPAN}{\mathop{\mathrm{P}_\mathrm{spanning\ tree}}}
\newcommand{\PM}{\mathop{\mathrm{P}_\mathrm{perfect\ matching}}}

\newcommand{\conv}{\mathop{\mathrm{conv}}}
\newcommand{\nnegrk}{\mathop{\mathrm{rank}_+}} 

\newcommand{\xc}{\mathop{\mathrm{xc}}} 

\theoremstyle{plain}
\newtheorem{theorem}{Theorem}

\newtheorem{lemma}{Lemma}
\newtheorem{proposition}{Proposition}

\theoremstyle{definition}

\theoremstyle{remark}

\newtheorem{cor}[theorem]{Corollary}

\begin{document}

\title{Extended formulations, nonnegative factorizations, and randomized communication protocols\thanks{A previous and reduced version of this paper appeared in the Proceedings of ISCO 2012.}}

\author{Yuri Faenza \thanks{Institut de math\'ematiques d'analyse et applications, EPFL, Lausanne, Switzerland, {\em yuri.faenza@epfl.ch}. Supported by the German Research Foundation (DFG) within the Priority Programme 1307 Algorithm Engineering.}\and Samuel Fiorini\thanks{D\'epartement de Math\'ematique, Universit\'e Libre de Bruxelles CP 216, Boulevard du Triomphe, 1050 Brussels, Belgium, {\em sfiorini@ulb.ac.be}. Partially supported by the \emph{Actions de Recherche Concert\'ees} (ARC) fund of the French community of Belgium.} \and Roland Grappe\thanks{Laboratoire d'Informatique de Paris-Nord, UMR CNRS 7030, Institut Galil\'ee - Universit\'e Paris-Nord, Avenue Jean-Baptiste Cl\'ement, 93430 Villetaneuse, France, {\em roland.grappe@lipn.univ-paris13.fr}. Partially supported by the Progetto di Eccellenza
2008--2009 of the Fondazione Cassa di Risparmio di Padova e
Rovigo.} \and Hans Raj Tiwary\thanks{D\'epartement de Math\'ematique, Universit\'e Libre de Bruxelles CP 216, Boulevard du Triomphe, 1050 Brussels, Belgium. {\em htiwary@ulb.ac.be}. Postdoctoral Researcher of the \emph{Fonds National de la Recherche Scientifique} (F.R.S.--FNRS).}}


\maketitle

\begin{abstract}
An extended formulation of a polyhedron $P$ is a linear description of a polyhedron $Q$ together with a linear map $\pi$ such that $\pi(Q)=P$. These objects are of fundamental importance in polyhedral combinatorics and optimization theory, and the subject of a number of studies. Yannakakis' factorization theorem [M. Yannakakis. Expressing combinatorial optimization problems by linear programs.
{\em J. Comput. System Sci.}, 43(3):441--466 (1991)] provides a surprising connection between extended formulations and communication complexity, showing that the smallest size of an extended formulation of $P$ equals the nonnegative rank of its slack matrix $S$. Moreover, Yannakakis also shows that the nonnegative rank of $S$ is at most $2^c$, where $c$ is the complexity of any \emph{deterministic} protocol computing $S$. In this paper, we show that the latter result can be strengthened when we allow protocols to be \emph{randomized}. In particular, we prove that the base-$2$ logarithm of the nonnegative rank of any nonnegative matrix equals the minimum complexity of a randomized communication protocol computing the matrix in expectation. Using Yannakakis' factorization theorem, this implies that the base-$2$ logarithm of the smallest size of an extended formulation of a polytope $P$ equals the minimum complexity of a randomized 
communication protocol computing the slack matrix of $P$ in expectation. We show that allowing randomization in the protocol can be crucial for obtaining small extended formulations. Specifically, we prove that for the spanning tree and perfect matching polytopes, small variance in the protocol forces large size in the extended formulation.
\end{abstract}

%
%

\section{Introduction}

Extended formulations are a powerful tool for minimizing linear or, more generally, convex functions over polyhedra (see, e.g., Ziegler~\cite{Ziegler} for background on polyhedra and polytopes). Consider a polyhedron $P$ in $\mathbb{R}^d$ and a convex function $f : \mathbb{R}^d \to \mathbb{R}$, that has to be minimized over $P$. If a small size linear description of $P$ is known, then minimizing $f$ over $P$ can be done efficiently using an interior point algorithm, or the simplex algorithm if $f$ is linear and theoretical efficiency is not required.

However, $P$ can potentially have many facets. Or worse: it can be that no explicit complete linear description of $P$ is known. This does not necessarily make the given optimization problem difficult. A fundamental result of Gr\"otschel, Lov\'asz and Schrijver~\cite{GrotschelLovaszSchrijverBook}
states that if there exists an efficient algorithm solving the separation problem for $P$, then optimizing over $P$ can be done efficiently. However, this result uses the ellipsoid algorithm, which is not very efficient in practice. Thus it is desirable to avoid using the ellipsoid algorithm.

Now suppose that there exists a polyhedron $Q$ in a higher dimensional space $\mathbb{R}^e$ such that $P$ is the image of $Q$ under a linear projection $\pi : \mathbb{R}^e \to \mathbb{R}^d$. The polyhedron $Q$ together with the projection $\pi$ defines an {\DEF extension} of $P$, while we call {\DEF extended formulation} of $P$ any description of $Q$ by means of linear inequalities and equations, together with the map $\pi$. Minimizing $f$ over $P$ amounts to minimizing $f \circ \pi$ over $Q$. If $Q$ has few facets, then we can resort to an interior point algorithm or the simplex algorithm to solve the optimization problem. Of course, one should also take into account the size of the coefficients in the linear description of $Q$ and in the matrix of $\pi$. But this can essentially be ignored for $0/1$-polytopes $P$~\cite{Rothvoss11}.

The success of extended formulations is due to the fact that a moderate increase in dimension can result in a dramatic decrease in the number of facets. For instance, $P$ can have exponentially many facets, while $Q$ has only polynomially many. We will see examples of this phenomenon later in this paper. For more examples, and background, see the recent surveys by Conforti, Cornu\'ejols and Zambelli~\cite{ConfortiCornuejolsZambelli10}, Kaibel~\cite{Kaibel} and Wolsey~\cite{Wolsey11}.

Extensions provide an interesting measure of how ``complex'' a polyhedron is: define  the {\DEF size} of an extension $Q$ of $P$ as the number of facets of $Q$ and the {\DEF extension complexity} of a polyhedron $P$ as the minimum size of any extension of $P$. Following~\cite{FioriniKaibelPashkovichTheis11}, we denote this number by $\xc(P)$.
The {\DEF size} of an extended formulation of $P$ is the number of inequalities of the linear system (hence, neither equations nor variables are taken into account). Note that the size of an extended formulation is at least the size of the associated extension, and any extension $Q$ has an extended formulation describing $Q$ with the same size.

This paper builds on Yannakakis' seminal paper~\cite{Yannakakis91}. We briefly review his contribution, postponing formal definitions to Section \ref{sec:prelim}. Because we mainly consider polytopes, we assume from now on that $P$ is bounded, that is, $P$ is a polytope. (This is not a major restriction.)
Yannakakis' {\DEF factorization theorem} (Theorem \ref{thm:Y91}) states that to each size-$r$ extension of a polytope $P$ corresponds a {\DEF rank-$r$ nonnegative factorization} of some matrix $S(P)$ associated to $P$, called the {\DEF slack matrix}, and conversely to each rank-$r$ nonnegative factorization of $S(P)$ corresponds a size-$r$ extension of $P$. In particular, the extension complexity $\xc(P)$ equals the smallest rank of a nonnegative factorization of $S(P)$, that is, the {\DEF nonnegative rank} of $S(P)$. 

In~\cite{Yannakakis91}, Yannakakis also shows that every $\lg r$-complexity \emph{deterministic} protocol computing a nonnegative matrix $M$ determines a rank-$r$ nonnegative factorization of $M$.~\footnote{Throughout this paper, we use $\lg$ for binary logarithm.} By the aforementioned factorization theorem, this implies that one can produce extended formulations (and hence upper bounds to the extension complexity) via deterministic communication protocols. Yannakakis used this to obtain a quasipolynomial $n^{O(\log n)}$-size extension for the stable set polytope of a $n$-vertex perfect graph.

\paragraph{Our contribution} The main goal of this paper is to strengthen the connection between nonnegative rank of matrices (and hence, extension complexity of polytopes) and communication protocols. First we give a brief overview of our results and then provide more details along with an outline of the paper. Our contribution is threefold:

\begin{itemize}

\item We pinpoint the ``right'' model of communication protocol, that exactly corresponds to nonnegative factorizations. We remark that this was done independently by Zhang~\cite{Zhang10}. Proving such a correspondence is an important conceptual step since it gives a third equivalent way to think about extensions of polytopes, besides projections of polytopes and nonnegative factorizations. Communication protocols are very versatile and we hope that this paper will convince discrete optimizers to add this tool to their arsenal.

\item We provide examples of already known extensions, seen as communication protocols, and also of new extensions obtained from communication protocols.

\item We prove that the randomization allowed in our protocols is sometimes necessary for obtaining small size extensions. We give a general condition under which small variance in the protocol implies that the size of the corresponding extension is large, which in particular applies to the perfect matching polytope and spanning tree polytope. This indicates that Yannakakis' approach for the stable set polytope of a perfect graph \emph{cannot work} for the perfect matching polytope or spanning tree polytope, since his protocol is deterministic and hence the corresponding variance zero. 
\end{itemize}

More specifically, we define a new model of \emph{randomized communication protocols computing the matrix in expectation}. This generalizes the one used by Yannakakis in~\cite{Yannakakis91} (see Section \ref{sec:cc}; our definition differs substantially from the usual notion of of random protocol computing a matrix \emph{with high probability}, which can be found e.g. in~\cite{KushilevitzNisan97}). Our protocols perfectly model the relation between the nonnegative factorization of a matrix and communication complexity: in fact, we show that the base-$2$ logarithm of the nonnegative rank of any nonnegative matrix  (rounded up to the next integer) \emph{equals} the minimum complexity of a randomized communication protocol computing the matrix in expectation (Theorem \ref{thm:p_vs_f}). By Yannakakis' factorization theorem, this implies a new characterization of the extension complexity of polytopes (Corollary \ref{cor:protocol_extform}). 

We then provide evidence that these protocols are substantially more powerful than the deterministic ones used, e.g., by Yannakakis. In fact, one can associate to each protocol a \emph{variance} (see Subsection \ref{subsec:variance}) which, roughly speaking, indicates the ``amount of randomness'' of the protocol: protocols with variance zero are deterministic protocols. We show that no compact formulation for the spanning tree polytope arises from protocols with small variance (see Section \ref{sec:spanning_trees}), while we provide a randomized protocol that produces the $O(n^3)$ formulation for the spanning tree polytope of $K^n$ due to Martin~\cite{Martin91} (see Section \ref{subsec:MST}). 

We also investigate the existence of compact extended formulation for the matching polytope --- a fundamental open problem in polyhedral combinatorics. Yannakakis~\cite{Yannakakis91} (see also~\cite{KaibelPashkovichTheis10}) proved that every {\em symmetric} extension of the perfect matching polytope of the complete graph $K^n$ has exponential size (we do not formally define \emph{symmetric} here, since 
we shall not need it; the interested reader may refer to~\cite{Yannakakis91}).
We show that a negative result similar to the one of the spanning tree polytope holds true for matchings: no compact formulation for the matching polytope arises from protocols with small variance (see Section \ref{sec:lb_pm}). Thus, in particular, deterministic protocols cannot be used to provide compact extended formulations for the perfect matching polytope. We also provide a randomized protocol that produces a  $O(1.42^n)$ formulation for the matching polytope implicit in Kaibel, Pashkovich and Theis~\cite{KaibelPashkovichTheis10} (see Subsection \ref{subsec:pm_ub}). The negative results on both the spanning tree and the matching polytopes are obtained via a general technique that exploits known negative results on the communication complexity of the set disjointness problem.

We would like to remark that the results contained in this paper were, at a conceptual level, an important stepping stone for the strong lower bounds on the extension complexities of the cut, stable set and TSP polytopes of Fiorini, Massar, Pokutta, Tiwary and de Wolf~\cite{FMPTW11}.

\section{Preliminary definitions and results}
\label{sec:prelim}

\subsection{The factorization theorem and related concepts} 

Consider a polytope $P$ in $\mathbb{R}^d$ with $m$ facets and $n$ vertices. Let $h_1$, \ldots, $h_m$ be $m$ affine functions on $\mathbb{R}^d$ such that $h_1(x) \geqslant 0$, \ldots, $h_m(x) \geqslant 0$ are all the facet-defining inequalities of $P$. Let also $v_1$, \ldots, $v_n$ denote the vertices of $P$. The {\DEF slack matrix} of $P$ is the nonnegative $m \times n$ matrix $S = S(P) = (s_{ij})$ with $s_{ij} = h_i(v_j)$. Also note that the facet-defining inequalities can be defined up to any positive scaling factor. It should be clear that such a scaling does not alter the non-negative rank of a matrix. To see this let $S=AB$ and let $S'$ be a matrix obtained by multiplying the $i$-th row of $S$ by $\lambda > 0$. Then, $S'=A'B$ where $A'$ is obtained by multiplying the $i$-th row of $A$ by $\lambda$.

A {\DEF rank-$r$ nonnegative factorization} of a nonnegative matrix $S$ is an expression of $S$ as a product $S = AB$ where $A$ and $B$ are nonnegative matrices with $r$ columns and $r$ rows, respectively. The {\DEF nonnegative rank} of $S$, denoted by $\nnegrk(S)$, is the minimum nonnegative integer $r$ such that $S$ admits a rank-$r$ nonnegative factorization~\cite{CohenRothblum93}. Observe that the nonnegative rank of $S$ can also be defined as the minimum nonnegative integer $r$ such that $S$ is the sum of $r$ nonnegative rank-$1$ matrices.

In a seminal paper, Yannakakis~\cite{Yannakakis91} proved, among other things, that the extension complexity of a polytope is precisely the nonnegative rank of its slack matrix (see also~\cite{FioriniKaibelPashkovichTheis11}).

\begin{theorem}[Yannakakis' factorization theorem] 
\label{thm:Y91}
For all polytopes $P$ that are neither empty or a point,
$$
\xc(P) = \nnegrk(S(P)).
$$
\end{theorem}

Before going on, we sketch the proof of half of the theorem. Assuming $P = \{x \in \mathbb{R}^d : Ex \leqslant g\}$, consider a rank-$r$ nonnegative factorization $S(P) = FV$ of the slack matrix of $P$. Then it can be shown that $Q := \{(x,y) \in \mathbb{R}^{d+r} : Ex + Fy = g,\ y \geqslant \mathbf{0}\}$ is an extension of $P$. Notice that $Q$ has at most $r$ facets, and $r$ extra variables\footnote{The extended formulation for $Q$ given above potentially has a large number of equalities, but recall we only consider the number of inequalities in the size of the extended formulation. The reasons for this are twofold: first, one can ignore most of the equalities after picking a small number of linearly independent equalities; and second, our concern in this paper is mainly the \emph{existence} of certain extensions.}. Taking $r = \nnegrk(S(P))$ implies $\xc(P) \leqslant \nnegrk(S(P))$. Moreover, since $P$ is a polytope, one can also assume that $Q$ is bounded, as shown by the following lemma.

\begin{lemma}\label{lem:boundedExtension}
Let $P = \{x \in \mathbb{R}^d : Ex \leqslant g\}$ be a polytope, let $S(P) = FV$ be a rank-$r$ nonnegative factorization of the slack matrix of $P$ with $r := \nnegrk(S(P))$, and let $Q := \{(x,y) \in \mathbb{R}^{d+r} : Ex + Fy = g,\ y \geqslant \mathbf{0}\}$. Then $Q$ is bounded.
\end{lemma}

\begin{proof}
The polyhedron $Q$ is unbounded if and only if its recession cone $\mathrm{rec}(Q) = \{(x,y) \in \mathbb{R}^{d+r} : Ex + Fy = \mathbf{0}, y \geqslant \mathbf{0}\}$ contains some nonzero vector. Since $P$ is bounded and the image of $Q$ under the projection $(x,y) \mapsto x$ is $P$, we have $x =\mathbf{0}$ for every point $(x,y) \in \mathrm{rec}(Q)$. Therefore, $Q$ is unbounded if and only if the system $Fy = \mathbf{0}, y \geqslant \mathbf{0}$ has a solution $y \neq \mathbf{0}$. But any such $y$ represents $\mathbf{0}$ as a non-trivial conical combination of the column vectors of $F$. Since $F$ is nonnegative, this is only possible if one of the columns of $F$ is identically zero, which would contradict the minimality of $r$. \qed \end{proof}

\subsection{Polytopes relevant to this work}
\label{sec:polytopes}

Now we describe briefly various families of polytopes relevant to this paper. For a more detailed description of these polytopes, we refer the reader to Schrijver~\cite{SchrijverBookB03}.

Let $I$ be a finite ground set. The {\DEF characteristic vector} of a subset $J \subseteq I$ is the vector $\chi^J \in \mathbb{R}^I$ defined as
\[
  \chi^J_i = \left\{
  \begin{array}{l l}
    1 & \quad \text{if } i \in J\\
    0 & \quad \text{if } i \notin J
  \end{array} \right.
\]
for $i \in I$. For $x \in \mathbb{R}^I$, we let $x(J) := \sum_{i \in J} x_i$.

Throughout this section, $G = (V,E)$ denotes a (finite, simple, undirected) graph. For a subset of vertices $U\subseteq V$, we denote the edges of the subgraph induced by $U$ as $E[U].$ The {\DEF cut} defined by $U$, denoted as $\delta(U)$, is the set of edges of $G$ exactly one of whose endpoints is in $U$. That is,
\begin{eqnarray*}
E[U] &= &\{uv \in E : u\in U, v \in U\}, \text{ and}\\
\delta(U) &= &\{uv \in EÊ:Êu\in U, v \notin U\}.
\end{eqnarray*}

Later in this paper, we will often take $G$ to be the {\DEF complete graph} $K^n$ with vertex set $V(K^n) = [n] := \{1,\ldots,n\}$ and edge set $E(K^n) = \{ij : i, j \in [n], i \neq j\}$.

\subsubsection{Spanning Tree Polytope}
\label{sec:st_polytope}

A {\DEF spanning tree} of $G$ is a tree $T=(V(T),E(T))$ (i.e., a connected graph without cycles) whose set of vertices and edges respectively satisfy $V(T) = V$ and $E(T) \subseteq E$. The {\DEF spanning tree polytope} of $G$ is the convex hull of the characteristic vectors of the spanning trees of $G$, i.e.,
$$
\SPAN(G) = \conv\{\chi^{E(T)} \in \mathbb{R}^E : T \text{ spanning tree of } G\}.
$$
Edmonds~\cite{Edmonds71} showed that this polytope admits the following linear description (see also~\cite{SchrijverBookB03}, page 861):
$$\begin{array}{rcll}
x(E[U]) &\leqslant &|U|-1 & \quad \text{for nonempty } U\subsetneq V,\\
x(E) &= &|V|-1,\\
x_e &\geqslant& 0 &  \quad \text{for } e\in E.
\end{array}
$$
This follows, e.g., from the fact that the spanning tree polytope of $G$ is the base polytope of the graphic matroid of $G$.

\subsubsection{Perfect Matching polytope}
\label{sec:pm_polytope}

A {\DEF perfect matching} of $G$ is set of edges $M \subseteq E$ such that every vertex of $G$ is incident to exactly one edge in~$M$. The {\DEF perfect matching polytope} of the graph $G$ is the convex hull of the characteristic vectors of the perfect matchings of $G,$ i.e.,
$$
\PM(G) = \conv\{\chi^M \in\mathbb{R}^E : M~ \text{perfect matching of}~ G\}.
$$
Edmonds~\cite{Edmonds65} showed that the perfect matching polytope of $G$ is described by the following linear constraints (see also~\cite{SchrijverBookB03}, page 438):
\begin{eqnarray*}
x(\delta(U)) &\geqslant &1 \quad \text{for } U\subseteq V \text{ with } |U| \text{ odd},\ |U| \geqslant 3\\
x(\delta(\{v\})) &= &1 \quad \text{for } v \in V,\\
x_e &\geqslant& 0 \quad \text{for } e\in E.
\end{eqnarray*}

\subsubsection{Stable Set polytope}
\label{sec:stab_polytope}

A {\DEF stable set} $S$ (often also called an {\DEF independent set}) of $G$ is a subset of the vertices such that no two of them are adjacent. A {\DEF clique} $K$ of $G$ is a subset of the vertices such that every two of them are adjacent. The {\DEF stable set polytope} $\STAB(G)$ of a graph $G(V,E)$ is the convex hull of the characteristic vectors of the stable sets in $G$, i.e.,
$$
\STAB(G) = \conv\{\chi^S \in\mathbb{R}^V : S \text{ stable set of } G\}.
$$

No complete linear description of the stable set polytope for arbitrary graphs is known. It is, however, known that the following inequalities are valid for $\STAB(G)$ for any graph $G$:
\begin{eqnarray}
\label{eq:stab_clique} x(K) &\leqslant &1 \quad\text{for cliques } K \text{ of } G,\\
\label{eq:stab_nneg} x_v &\geqslant &0 \quad\text{for } v\in V.
\end{eqnarray}
Inequalities \eqref{eq:stab_clique} are called the {\DEF clique inequalities}. See Schrijver~\cite{SchrijverBookB03} for details.

A graph $G$ is called {\DEF perfect} if the chromatic number of every induced subgraph equals the size of the largest clique of that subgraph. It is known that $G$ is perfect if and only if inequalities \eqref{eq:stab_clique} and \eqref{eq:stab_nneg} completely describe $\STAB(G)$~\cite{Chvatal75}.

\section{Communication complexity}
\label{sec:cc}

We start by an overview of the standard model of deterministic communication protocols, as described in detail in the book by Kushilevitz and Nisan~\cite{KushilevitzNisan97}. We follow this with a detailed description of our notion of a randomized protocol (with private random bits and nonnegative outputs) computing a function \emph{in expectation}. This differs significantly from the standard definition in the literature where randomized protocols usually compute a function exactly \emph{with high probability}. 

\subsection{Deterministic protocols}
\label{sec:det_protocols}

Let $X$, $Y$, and $Z$ be arbitrary finite sets with $Z \subseteq \mathbb{R}_+$, and let $f : X \times Y \rightarrow Z$ be a function. Suppose that there are two players Alice and Bob who wish to compute $f(x,y)$ for some inputs $x \in X$ and $y \in Y$. Alice knows only $x$ and Bob knows only $y$. They must therefore exchange information to be able to compute $f(x,y)$. (We assume that each player possesses unlimited computational power.)

The communication is carried out as a protocol that is agreed upon beforehand by Alice and Bob, on the sole basis of the function $f$. At each step of the protocol, one of the player has the token. Whoever has the token sends a bit to the other player, that depends only on their input and on previously exchanged bits. 
This is repeated until the value of $f$ on $(x,y)$ is known to both players. The minimum number of bits exchanged between the players in the worst case to be able to evaluate $f$ by any protocol is called the \emph{communication complexity} of $f.$

\subsection{Randomized protocols and computation in expectation}
\label{sec:rand_protocols}

A protocol can be viewed as a rooted binary tree where each node is marked either Alice or Bob. The leaves have vectors associated with them. An execution of the protocol on a particular input is a path in the tree starting at the root. At a node owned by Alice, following the path to the left subtree corresponds to Alice sending a zero to Bob and taking the right subtree corresponds to Alice sending a one to Bob; and similarly for nodes owned by Bob. 

More formally, we define a {\DEF randomized protocol (with private random bits and nonnegative outputs)} as a rooted binary tree with some extra information attached to its nodes. Let $X$ and $Y$ be finite sets, as above. Each node of the tree has a {\DEF type}, which is either $X$ or $Y$. To each node $v$ of type $X$ are attached two function $p_{0,v}, p_{1,v} : X \to [0,1]$; to each node $v$ of type $Y$ are attached two functions $q_{0,v}, q_{1,v} : Y \to [0,1]$; and to each leaf $v$ is attached a nonnegative vector $\Lambda_v$ that is a column vector of size $|X|$ for leaves of type $X$ and a row vector of size $|Y|$ for leaves of type $Y$. The functions $p_{i,v}$ and $q_{j,v}$ define {\DEF transition probabilities}, and we assume that $p_{0,v}(x)+p_{1,v}(x) \leqslant 1$ and $q_{0,v}(y)+q_{1,v}(y) \leqslant 1$. Figure \ref{fig:protocol_example} shows an example of a protocol.

\begin{figure}[ht]
  \centering
  \begin{subfigure}[b]{0.35\textwidth}
    \centering
    \includegraphics[width=\textwidth]{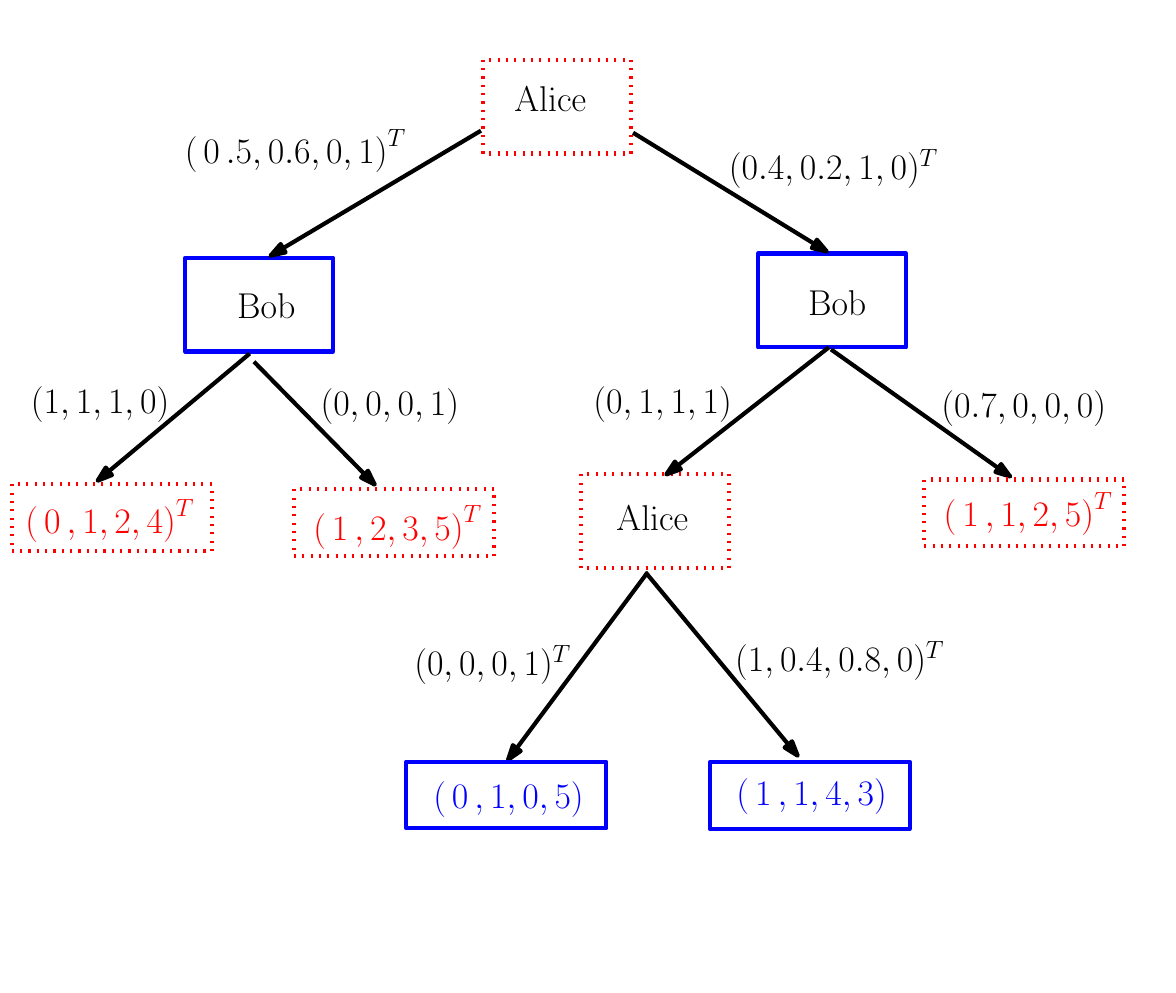}
    \caption{Protocol as a tree.}
  \end{subfigure}
  \begin{subfigure}[b]{0.35\textwidth}
    \centering
    \includegraphics[width=\textwidth]{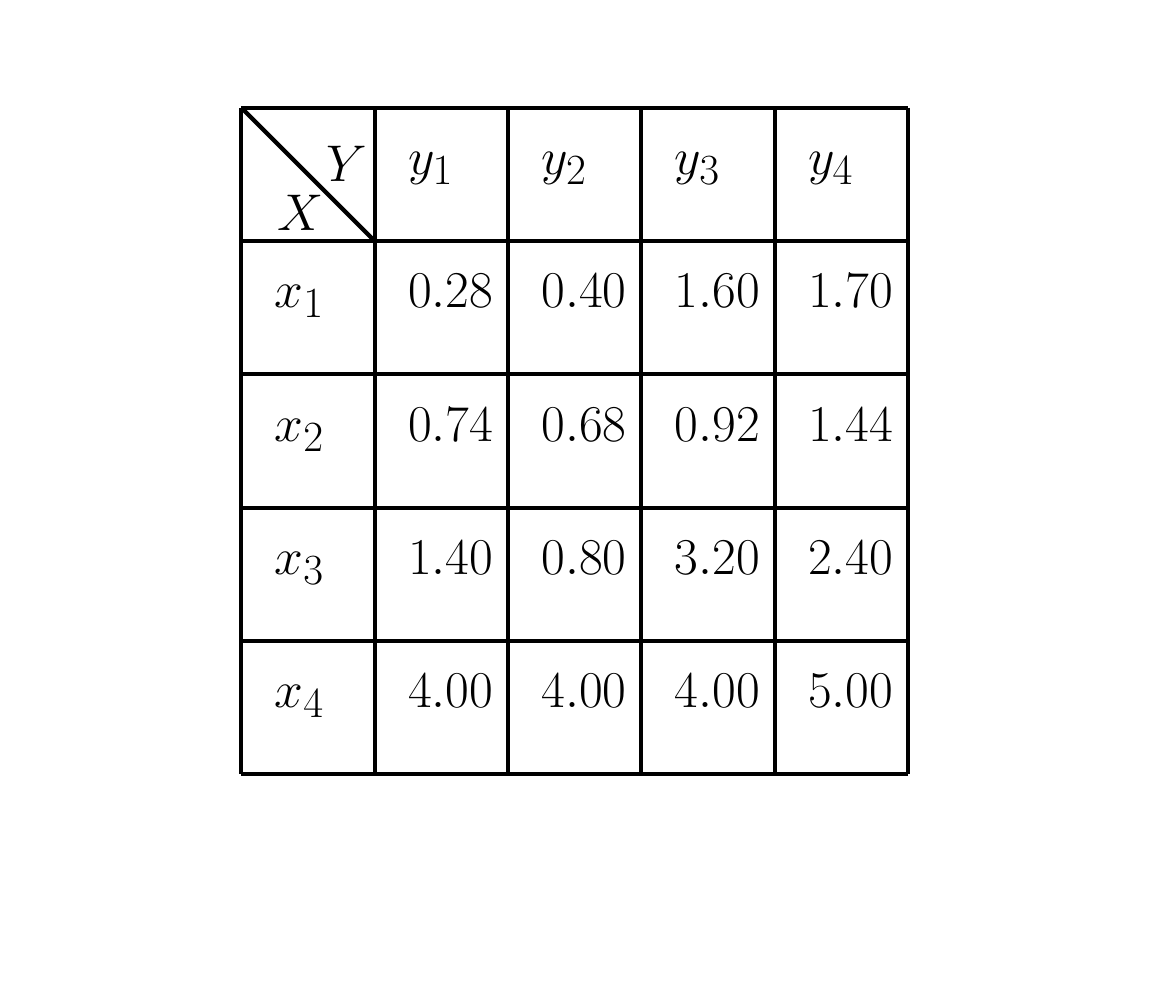}
    \caption{The associated communication matrix.}
  \end{subfigure}
  \caption{Illustration of a (non-optimal) randomized protocol computing a matrix in expectation.}
  \label{fig:protocol_example}
\end{figure}

An {\DEF execution} of the protocol on input $(x,y) \in X \times Y$ is a random path that starts at the root and descends to the left child of an internal node $v$ with probability $p_{0,v}(x)$ if $v$ is of type $X$ and $q_{0,v}(y)$ if $v$ is of type $Y$, and to the right child of $v$ with probability $p_{1,v}(x)$ if $v$ is of type $X$ and $q_{1,v}(y)$ if $v$ is of type $Y$. With probability $1-p_{0,v}(x)-p_{1,v}(x)$ and $1-q_{0,v}(y)-q_{1,v}(y)$ respectively, the execution stops at $v$. For an execution stopping at leaf $v$ with vector $\Lambda_v$, the {\DEF value} of the execution is defined as the entry of $\Lambda_v$ that corresponds to input $x \in X$ if $v$ is of type $X$, and $y \in Y$ if $v$ is of type $Y$. For an execution stopping at an internal node, the {\DEF value} is defined to be $0$.

For each fixed input $(x,y) \in X \times Y$, the value of an execution on input $(x,y)$ is a random variable. If we let $Z\subseteq \mathbb{R}_+$ as before, we say that the protocol {\DEF computes} a function $f : X \times Y \to Z$ {\DEF in expectation} if the expectation of this random variable on each $(x,y) \in X \times Y$ is precisely $f(x,y)$.

The {\DEF complexity} of a protocol is the height of the corresponding tree. 

Given an ordering $x_1$, \ldots, $x_m$ of the elements of $X$, and $y_1$, \ldots, $y_n$ of the elements of $Y$, we can visualize the function $f : X \times Y \to Z$ as a $m \times n$ nonnegative matrix $S = S(f)$ such that $S_{i,j} = f(x_i,y_j)$ for all $(i,j) \in [m] \times [n]$. The matrix $S$ is called the {\DEF communication matrix} of $f$. Below, as is natural, we will not always make a distinction between a function and its communication matrix. 

These formal definitions capture the informal ones given above. Observe that the nodes of type $X$ are assigned to Alice, and those of type $Y$ to Bob. Observe also that Alice and Bob have unlimited resources for performing their part of the computation. It is only the communication between the two players that is accounted for. When presenting a protocol, we shall often say that one of the two players sends an integer $k$ rather than a binary value. This should be interpreted as the player sending the binary encoding of $k$ or, as a (sub)tree of height $\lceil \lg k \rceil$. Finally, our definitions are such that the complexity of a protocol equals the number of bits exchanged by Alice and Bob. 

\subsection{Normalized Variance}\label{subsec:variance}
Since the output of a randomized protocol - as defined above - is a random variable, one can define its variance. However, we would like to refine the notion of variance so that protocols computing different scalings of the same matrix have the same variance. This is essential since the nonnegative rank of a matrix is an invariant under scaling and, as we will see in the next section, there is an equivalence between the nonnegative rank of a matrix $S$ and the smallest complexity protocol computing $S$ in expectation.

Let $S$ be a nonnegative matrix and suppose there exists a protocol of complexity $c$ computing $S$ in expectation. Let $\xi_{i,j}$ denote the random variable corresponding to the output of the protocol on input $(x_i,y_j) \in X \times Y.$ That is $\mathbb{E}[\xi_{i,j}]=S_{i,j}.$ The  \emph{normalized variance} $\sigma^2$ of the protocol is defined as the maximum variance of the random variables $\xi'_{i,j}=\frac{\xi_{i,j}}{S_{i,j}}$ for the nonzero entries of $S.$ That is $$\sigma^2= \max_{(i,j) \mid S_{i,j}\neq 0}\text{Var}(\xi_{i,j}/S_{i,j}) $$

\section{Factorizations vs. protocols}
\label{sec:p2f}

\begin{theorem}
\label{thm:p_vs_f}
If there exists a randomized protocol of complexity $c$ computing a matrix $S \in \mathbb{R}^{X \times Y}_+$ in expectation, then $\lg \nnegrk(S) \leqslant c$. Conversely, if the nonnegative rank of matrix $S \in \mathbb{R}^{m \times n}_+$ is $r$, then there exists a randomized protocol computing $S$ in expectation, whose complexity is at most $\lceil \lg r \rceil$. In other words, if $c_{\mathrm{min}}(S)$ denotes the minimum complexity of a randomized protocol computing $S$ in expectation, we have
$$
c_{\mathrm{min}}(S) = \lceil \lg \nnegrk(S) \rceil.
$$
\end{theorem}

\begin{proof}
Suppose there exists a protocol of complexity $c$ computing $S$ in expectation. Each node $v$ of the protocol has a corresponding {\DEF traversal probability matrix} $P_v \in \mathbb{R}^{X \times Y}_+$ such that, for all inputs $(x,y) \in X \times Y$, the entry $P_v(x,y)$ is the probability that an execution on input $(x,y)$ goes through node $v$. 

Let $v_1$, \ldots, $v_k$ denote the nodes of type $X$ on the unique path from the root to the parent of $v$, and let $w_1$, \ldots, $w_\ell$ denote the nodes of type $Y$ on this path. Then we have
$$
P_v(x,y) = \prod_{i=1}^k p_{\alpha_i,v_i}(x) \cdot \prod_{j=1}^\ell q_{\beta_j,w_j}(y),
$$
where $\alpha_i$ is either $0$ or $1$ depending on if the path goes the left or right subtree at $v_i$, and similarly for $\beta_j$. We immediately see that $P_v$ is a rank one matrix of the form $a_v b_v$ where $a_v$ is a column vector of size $|X|$ and $b_v$ is a row vector of size $|Y|$.

Finally, let $L_X$ and $L_Y$ be the set of all leaves of the protocol that are of type $X$ and $Y$ respectively and let $\Lambda_v$ denote the (column or row) vector of values at a leaf $v \in L_X \cup L_Y$. Because the protocol computes $S$ in expectation, for all inputs $(x,y) \in X \times Y$ we have $S(x,y) = \sum_{v \in L_X} \Lambda_v(x) P_v(x,y)+\sum_{w \in L_Y} P_w(x,y)\Lambda_w(y) $. Thus, $S = \sum_{v \in L_X} (\Lambda_v\circ a_v)b_v + \sum_{v \in L_Y}  a_w(b_w \circ \Lambda_w),$ where $\circ$ denotes the Hadamard product. Therefore, it is possible to express $S$ as a sum of at most $|L_X \cup L_Y| \leqslant 2^c$ nonnegative rank one matrices. Hence, $\nnegrk(S) \leqslant 2^c$, that is, $\lg \nnegrk(S) \leqslant c$

To prove the other part of the theorem, let $A \in \mathbb{R}_+^{m \times r}$ and $B \in \mathbb{R}_+^{r \times n}$ be nonnegative matrices such that $S = AB$. By scaling, we can assume that the maximum row sum of $A$ is $1$. Otherwise, we replace $A$ and $B$ by $\Delta^{-1}A$ and $\Delta B$ respectively, where $\Delta$ denotes the maximum row sum of $A$. 

The protocol is as follows: Alice knows a row index $i$, and Bob knows a column index $j$. Together they want to compute $S_{i,j}$ in expectation, by exchanging as few bits as possible. They proceed as follows. Let $\delta_i := \sum_{k} A_{i,k} \leqslant 1.$  Alice selects a column index $k \in [r]$ according to the probabilities found in row $i$ of matrix $A$, sends this index to Bob, and Bob outputs the entry of $B$ in row $k$ and column $j$. With probability $1-\delta_i$ Alice does not send any index to Bob and the computation stops with implicit output zero (see subsection \ref{sec:rand_protocols}).

This randomized protocol computes the matrix $S$ in expectation. Indeed, the expected value on input $(i,j)$ is $\sum_{k=1}^{r} A_{i,k} B_{k,j} = S_{i,j}$. Moreover, the complexity of the protocol is precisely $\lceil \lg (r) \rceil$. \qed
\end{proof}

The above theorem together with Theorem \ref{thm:Y91} gives us the following corollary:

\begin{cor}\label{cor:protocol_extform}
Let $P$ be a polytope with associated slack matrix $S = S(P)$, such that $P$ is neither empty or a point. If there exists a randomized protocol of complexity $c$ computing $S$ in expectation, then $\xc(P)\leqslant 2^c$. Conversely, if $\xc(P)=r$, then there exists a randomized protocol computing $S$ in expectation, whose complexity is at most $\lceil \lg r \rceil$. In other words, if $c_{\mathrm{min}}(S)$ denotes the minimum complexity of a randomized protocol computing $S$ in expectation, we have
$$
c_{\mathrm{min}}(S(P)) = \lceil \lg \xc(P) \rceil.
$$
\end{cor}

The concrete polytopes considered in this paper have some facet-defining inequalities enforcing nonnegativity of the variables along with other facet-defining inequalities. The next lemma and its corollary will allow us to ignore the rows corresponding to nonnegativity inequalities, and focus on the non-trivial parts of the slack matrices.

\begin{lemma}
\label{lem:partition}
Let $S$ be a nonnegative matrix. Let $R_1, R_2$ be a partition of the rows of $S$ defining partition of $S$ into $S_1$ and $S_2$. If there exist randomized protocols computing $S_1$ and $S_2$ in expectation with complexity $c_1$ and $c_2$ respectively, then there exists a randomized protocol complexity computing $S$ with complexity $1+\max\{c_1,c_2\}$.
\end{lemma}
\begin{proof}
When Alice gets a row index of $S$ she sends a bit to Bob to indicate whether the corresponding row lies in $R_1$ or $R_2$. Now that both Alice and Bob know whether they want to compute an entry in $S_1$ or $S_2$, they use the protocol for that particular submatrix. \qed
\end{proof}

\begin{cor}
\label{cor:split}
Let $P \subseteq \mathbb{R}_+^d$ be a polytope and let $S'(P)$ denote the submatrix of $S(P)$ obtained by deleting the rows corresponding to nonnegativity inequalities. If there is a complexity $c$ randomized protocol for computing $S'(P)$ in expectation, then there is a complexity $1+\max\{c,\lceil\lg d\rceil\}$ randomized protocol for computing $S(P)$ in expectation.
\end{cor}
\begin{proof}
For computing the part of $S(P)$ that is deleted in $S'(P)$, which corresponds to nonnegativity inequalities, we use the obvious protocol where Alice sends her row number to Bob and Bob computes the slack. Since at most $d$ facets of $P$ are defined by nonnegativity inequalities, this protocol has complexity $\lceil\lg d\rceil$. The corollary thus follows from Lemma~\ref{lem:partition}. \qed
\end{proof}

For the protocols constructed here, we will always have $c \geqslant \lceil\lg d\rceil$. Because of  Corollary~\ref{cor:split}, we can thus ignore the nonnegativity inequalities without blowing up the size of any extension by more than a factor of $2$. Moreover, in terms of lower bounds, it is always safe to ignore inequalities because the nonnegative rank of a matrix cannot increase when rows are deleted.

\section{Examples}

In this section, we give three illustrative examples of protocols defining nonnegative factorizations of various slack matrices, and thus (via Corollary~\ref{cor:protocol_extform}) extensions of the corresponding polytopes. The first one gives a $O(n^3)$-size extension of the stable set polytope of a claw-free perfect graph. The second one is a reinterpretation of a well-known $O(n^3)$-size extended formulation for the spanning tree polytopes due to Martin~\cite{Martin91}. Our interpretation allows for a more general result. In particular we prove new upper bounds for the spanning tree polytopes for minor-free graphs. The third one concerns the perfect matching polytopes and is implicit in Kaibel, Pashkovich and Theis~\cite{KaibelPashkovichTheis10}.

\subsection{The stable set polytope of a claw-free perfect graph}
\label{ex:claw-free_perfect}

A graph $G$ is called {\DEF claw-free} if no vertex has three pairwise non-adjacent neighbors. Even though the separation problem for $\STAB(G)$ for claw-free graphs is polynomial-time solvable, no explicit description of all its facets is known (see, e.g.,~\cite{SchrijverBookB03}, page 1216). Recently Faenza, Oriolo, and Stauffer~\cite{FaenzaOrioloStauffer12} provided (non-compact) extended formulations for this polytope, while Galluccio et al.~\cite{GalluccioGentileVentura10} gave a complete description of the facets for claw-free graphs with at least one stable set of size greater than or equal to four, and no clique-cutsets.
Also, recall that for a perfect graph $G$ the facets of $\STAB(G)$ are defined by inequalities (\ref{eq:stab_clique}) and (\ref{eq:stab_nneg}) (see Section~\ref{sec:stab_polytope}).

Let $G$ be a claw-free, perfect graph with $n$ vertices. We give a deterministic protocol that computes the slack matrix of the stable set polytope $\STAB(G)$ of $G$. Because $G$ is perfect, the (non-trivial part of the) slack matrix of $\STAB(G)$ has the following structure: it has one column per stable set $S$ in $G$, and each one of its rows corresponds to a clique $K$ in $G$. The entry for a pair $(K,S)$ equals $0$ if $K$ and $S$ intersect (in which case they intersect in exactly one vertex) and $1$ if $K$ and $S$ are disjoint (note that we are ignoring the $|V|$ rows that correspond to nonnegativity inequalities (\ref{eq:stab_nneg}). This can be done safely, see Corollary~\ref{cor:split}).

Consider the communication problem in which Alice is given a clique $K$ of $G$, Bob is given a stable set $S$ of $G$, and Alice and Bob together want to compute $1 - |K \cap S|$. Alice starts and sends the name of any vertex $u$ of her clique $K$ to Bob. Then Bob sends the names of all the vertices of its stable set $S$ that are in $N(u) \cup \{u\}$ to Alice, where $N(u)$ denotes the neighborhood of $u$ in $G$. Finally, Alice can compute $K \cap S$ because this intersection is contained in $N(u) \cup \{u\}$ and Alice knows all vertices of $S \cap (N(u) \cup \{u\})$. She outputs $1 - |K \cap S|$. Because $G$ is claw-free, there are at most two vertices in $S \cap (N(u) \cup \{u\})$, thus at most $3\lg n + O(1)$ bits are exchanged by Alice and Bob. It follows that there exists an extension (and hence, an extended formulation) of $\STAB(G)$ of size $O(n^3)$. Notice that the normalized variance of our protocol is zero, because it is deterministic.

We obtain the following result.

\begin{proposition}
 For every perfect, claw-free graph $G$ with $n$ vertices, $\STAB(G)$ has an extended formulation of size $O(n^3)$.
\end{proposition}

\subsection{The spanning tree polytope}\label{subsec:MST} 

Let $\SPAN(G)$ denote the spanning tree polytope of a graph $G=(V,E)$ (see Section~\ref{sec:st_polytope}). The (non-trivial part of the) slack matrix of $P$ has one column per spanning tree $T$ and one row per proper nonempty subset $U$ of vertices. The slack of $T$ with respect to the inequality that corresponds to $U$ is the number of connected components of the subgraph of $T$ induced by $U$ (denoted by $T[U]$ below) minus one.

In terms of the corresponding communication problem, Alice has a proper non\-empty set $U$ and Bob a spanning tree $T$. Together, they wish to compute the slack of the pair $(U,T)$. Alice sends the name of some (arbitrarily chosen) vertex $u$ in $U$. Then Bob picks an edge $e$ of $T$ uniformly at random and sends to Alice the endpoints $v$ and $w$ of $e$ as an ordered pair of vertices $(v,w)$, where the order is chosen in such a way that $w$ is on the unique path from $v$ to $u$ in the tree. That is, she makes sure that the directed edge $(v,w)$ ``points'' towards the root $u$. Then Alice checks that $v \in U$ and $w \notin U$, in which case she outputs $n-1$; otherwise she outputs $0$.

\begin{figure}[ht]
\centering
\scalebox{.75}{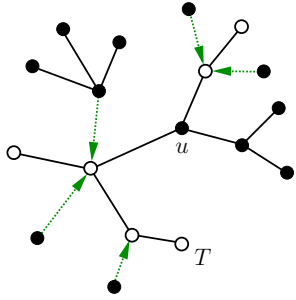}
\caption{Illustration of the protocol for the slack of MST polytope. The black vertices are those in $U$. The green directed edges are those for which Alice outputs a non-zero value. The number of such edges is the number of connected components of $T[U]$ minus one.}
\label{fig:tree}
\end{figure}

The resulting randomized protocol is clearly of complexity $\lg{|V|} + \lg{|E|} + O(1)$. Moreover, it computes the slack matrix in expectation because for each connected component of $T[U]$ distinct from that which contains $u$, there is exactly one directed edge $(v,w)$ that will lead Alice to output a non-zero value, see Figure~\ref{fig:tree} for an illustration. Since she outputs $(n-1)$ in this case, the expected value of the protocol on pair $(U,T)$ is $(n-1) \cdot (k-1)/(n-1) = k-1$, where $k$ is the number of connected components of $T[U]$. Therefore, we obtain the following result. 

\begin{proposition}
 For every graph $G$ with $n$ vertices and $m$ edges, $\SPAN(G)$ has an extended formulation of size $O(mn).$
\end{proposition}

The above result is implicit in Martin~\cite{Martin91}, although the paper only states the following corollary. More specifically, variables $z_{i,j,k}$ such that $ij$ is not an edge of $G$ can be deleted from his $O(n^3)$-size extended formulation, so that the resulting formulation has size $O(mn)$.

\begin{cor}
 $\SPAN(K^n)$ has extended formulation of size $O(n^3)$, where $K^n$ is the complete graph on $n$ vertices.
\end{cor}

\begin{cor} \label{cor:minor-free}
 Let $G$ be an $H$ minor-free graph, where $H$ is a graph with $h$ vertices, then $\SPAN(G)$ has extended formulation of size $O(n^2h\sqrt{\lg h}).$
\end{cor}
\begin{proof}
 It is known that any $H$ minor-free graph $G$ with $n$ vertices has at most $O(nh\sqrt{\lg h})$ edges, where $h$ is the number of vertices of $H$~\cite{Thomason01}. The result follows. \qed
\end{proof}

We remark that when $G$ is planar, $\SPAN(G)$ has an extended formulation of size $O(n)$~\cite{Williams02}. It is natural to ask whether a linear size extended formulation also exists for general $H$ minor-free graphs. So far, the best that seems to be known is the upper bound in Corollary \ref{cor:minor-free}.

Finally, it can be easily verified that the normalized variance of the protocol given above is $\sigma^2 = n-2$, which is large compared to the previous protocol. 

\subsection{Perfect Matching Polytope}\label{subsec:pm_ub}
For the next example, we will need the fact that one can cover $K^n$ with $k = O(2^{n/2} \mathrm{poly}(n))$ balanced complete bipartite graphs $G_1$,\ldots, $G_k$ in such a way that every perfect matching of $K^n$ is a perfect matching of at least one of the $G_i$'s. We say that $X \subseteq [n]$ is an \emph{$(n/2)$-subset} of $[n]$ if $|X|=n/2$. Given a matching $M$ of $K^n$ and a $(n/2)$-subset $X$ of $[n]$, we say that $X$ is \emph{compatible} with $M$ if all the edges of $M$ have exactly one end in $X$. 

\begin{lemma}
\label{lem:pm_cover}
Let $n$ be an even positive integer. There exists a collection of $k = O(2^{n/2} \sqrt{n} \ln n)$ $(n/2)$-subsets $X_1$,\ldots, $X_k$ of $[n]$ such that for every perfect matching $M$ of $K^n$ at least one of the subsets $X_i$ is compatible with $M$.
\end{lemma}
\begin{proof}
Finding a minimum size such collection $X_1$, \ldots, $X_k$ amounts to solving a set covering instance that we formulate by an integer linear program. For each $(n/2)$-subset $X$, we define a variable binary variable $\lambda(X)$. For each each perfect matching $M$, these variables have to satisfy the constraint $\sum \{\lambda(X) : X$ is compatible with $M\} \geqslant 1$. The goal is to minimize $\sum \lambda(X)$, the sum of all variables $\lambda(X)$.

A feasible fractional solution to this linear program is to let $\lambda^*(X) = 1/2^{n/2}$. This gives a feasible fractional solution because each perfect matching $M$ is compatible with exactly $2^{n/2}$ $(n/2)$-subsets $X$, so $\sum \{\lambda^*(X) : X$ is compatible with $M\} = 2^{n/2} (1/2^{n/2}) = 1$. (By symmetry considerations, it is in fact possible to argue that this solution is actually optimal.) The cost of this fractional solution $\lambda^*$ is
$$
\sum \lambda^*(X) = \frac{1}{2^{n/2}} {n \choose n/2} \leqslant \frac{2^{n/2}}{\sqrt{n}},
$$
for $n$ sufficiently large. By Lov\'asz's analysis of the greedy algorithm for the set covering problem~\cite{Lovasz_greedy}, there exists a feasible integer solution $\lambda$ of cost at most $(1 + \ln u)$ times the fractional optimum, where $u$ is the number of elements to cover. By what precedes, this is at most
$$
\left(1 + \ln \frac{n!}{2^{n/2} (n/2)!}\right) \frac{2^{n/2}}{\sqrt{n}} = O(2^{n/2} \sqrt{n} \lg n),
$$
from which the result follows directly. \qed
\end{proof}

Assume that $n$ is even and let $P$ denote the perfect matching polytope of the complete graph $K^n$ with vertex set $[n]$, see Section~\ref{sec:pm_polytope}. The (non-trivial part of the) slack matrix of $P$ has one column per perfect matching $M$, and its rows correspond to odd sets $U \subseteq [n]$. The entry for a pair $(U,M)$ is $|\delta(U) \cap M| - 1$ (recall that $\delta(U)$ denotes the set of edges that have one endpoint in $U$ and the other endpoint in $\overline U$, the complement of $U$).

We describe a randomized protocol for computing the slack matrix in expectation, of complexity at most $(1/2 + \varepsilon)n$, where $\varepsilon > 0$ can be made as small as desired by taking $n$ large. First, Bob finds an $(n/2)$-subset $X \subseteq [n]$ that is compatible with his matching $M$, and tells the name of this subset to Alice, see Lemma \ref{lem:pm_cover}. Then Alice checks which of $X$ and $\overline{X}$ contains the least number of vertices of her odd set $U$. Without loss of generality, assume it is $X$. If $U \cap X = \varnothing$ then, because $U \subseteq \bar{X}$ and $X$ is compatible with $M$, Alice can correctly infer that the slack is $|U| - 1$, and outputs this number. Otherwise, she picks a vertex $u$ of $U \cap X$ uniformly at random and send its name to Bob. He replies by sending the name of $u'$, the mate of $u$ in the matching $M$. Alice then checks whether $u'$ is in $U$ or not. If $u'$ is not in $U$, then she outputs $|U| - 1$. Otherwise $u'$ is in $U$, and she outputs $|U| - 
1 - 2|U \cap X|$. Telling the name of $X$ can be done in at most $n/2 + \lg \sqrt{n} + \lg \lg n + O(1)$ bits, see Lemma \ref{lem:pm_cover}. The extra amount of communication is $2 \lg n + O(1)$ bits. In total, at most $(1/2 + \varepsilon)n$ bits are exchanged, for $n$ sufficiently large ($\varepsilon > 0$ can be chosen arbitrarily).

Now, we check that the protocol correctly computes the slack matrix of the perfect matching polytope. Letting $E[U]$ denote the edges of the complete graph with both endpoints in $U$, the expected value output by Alice (in the case $U \cap X \neq \varnothing$) is
\begin{eqnarray*}
&&(|U| - 1) \frac{|U \cap X| - |E[U] \cap M|}{|U \cap X|} +
(|U|- 1 - 2 |U \cap X|) \frac{|E[U] \cap M|}{|U \cap X|}\\
&=& |U| - 1 - 2 |U \cap X | \frac{|E[U] \cap M|}{|U \cap X|}\\
&=& |U| - 2 |E[U] \cap M| - 1\\
&=& |\delta(U) \cap M| - 1.
\end{eqnarray*}

We obtain the following result.

\begin{proposition}
 Let $\varepsilon > 0$. For every large enough even nonnegative integer $n$, $\PM(K^n)$ has an extended formulation of size at most $2^{(1/2+\varepsilon)n}$.
\end{proposition}

We remark that our extension has size at most $2^{(1/2+\varepsilon)n} \leqslant (1.42)^n$, whereas the main result of Yannakakis~\cite{Yannakakis91} gives a lower bound of ${n \choose n/4} \geqslant (1.74)^n$ for the size of any \emph{symmetric} extension. 

\section{When low variance forces large size}
\label{sec:lb_pm}

We have seen that every extension of a polytope $P$ corresponds to a randomized protocol computing its slack matrix $S = S(P)$ in expectation and vice-versa. Now we show that if the set disjointness matrix can be embedded in a certain way in a matrix $S$ (see below for definitions), then efficient protocols computing $S$ in expectation necessarily have large variance. We prove that such an embedding can be found for the slack matrices of the perfect matching polytope and also, surprisingly, of the spanning tree polytope.

\subsection{Embedding the set disjointness matrix} 

The {\DEF set disjointness problem} is the following communication problem: Alice and Bob each are given a subset of $[n]$. They wish to determine whether the two subsets intersect or not. In other words, Alice and Bob have to compute the {\DEF set disjointness matrix} $\textrm{DISJ}$ defined by $\textrm{DISJ}(A,B) = 1$ if $A$ and $B$ are disjoint subsets of $[n]$, and $\textrm{DISJ}(A,B) = 0$ if $A$ and $B$ are non-disjoint subsets of $[n]$. The set disjointness problem plays a central role in communication complexity, comparable to the role played by the satisfiability problem in NP-completeness theory~\cite{ChattopadhyayPitassi2010}. 

It is known that any randomized protocol that computes the disjointness function {\em with high probability} (that is, the probability that the value output by the protocol is correct is, for each input, bounded from below by a constant strictly greater than $1/2$) has $\Omega(n)$ complexity~\cite{KalyanasundaramSchnitger92,Razborov92}.

Consider a matrix $S \in \mathbb{R}^{X \times Y}_+$. An {\DEF embedding} of the set disjointness matrix on $[n]$ in $S$ is defined by two maps $\alpha : 2^{[n]} \to X$ and $\beta : 2^{[n]} \to Y$ such that 
\begin{equation}
\label{eq:embedding}
\forall A, B \subseteq [n]: \textrm{DISJ}(A,B) = 1 \iff S(\alpha(A),\beta(B)) = 0.
\end{equation}
Notice that this kind of embedding could be called ``negative'' because zeros in the set disjointness matrix correspond to non-zeros in $S$. 

We remark that ``positive'' embeddings of the set disjointness matrix force up the rank of $S$, because the rank of any matrix with the same support as the set disjointness matrix on $[n]$ is at least $2^n$ \cite{HoyerWolf2002}. This is not desirable because the nonnegative rank of $S$ is always at least its rank. Thus the lower bound on the nonnegative rank of $S$ obtained from such a ``positive'' embedding would be useless in our context (the rank of the slack matrix $S(P)$ of polytope $P$ equals $\dim(P) + 1$).

However, ``positive'' embeddings the {\em unique set disjointness matrix}, that is the restriction of the set disjointness matrix to pairs $(A,B)$ such that $|A \cap B| \leqslant 1$, do not have this problem of forcing up the rank. Actually,  ``positive'' embeddings of the unique set disjointness matrix led to the main result of Fiorini \emph{et al.}~\cite{FMPTW11}.

\begin{theorem} \label{thm:conditional_lowerbound}
Let $S \in \mathbb{R}^{X \times Y}_+$ be a matrix in which the set disjointness matrix on $[n]$ can be embedded. Consider a randomized protocol computing $S$ in expectation. If the probability that the protocol outputs a non-zero value, given an input $(x,y)$ with $S(x,y) > 0$, is at least $p = p(n)$, then the protocol has complexity $\Omega(np)$. In particular, by Chebyshev's inequality, the complexity is $\Omega(n(1-\sigma^2))$, where
$\sigma^2$ denotes the normalized variance of the protocol. 
\end{theorem}
\begin{proof}
Let $c$ be the complexity of the protocol computing $S$ in expectation. From this protocol, we obtain a new protocol, this time for the set disjointness problem, by mapping each input pair $(A,B) \in 2^{[n]} \times 2^{[n]}$ to the corresponding input pair $(\alpha(A),\beta(B)) \in X \times Y$ (Alice and Bob can do this independently of each other), running the original protocol $\lceil1/p\rceil$ times, and outputting $0$ if at least one of the executions led to a non-zero value or $1$ otherwise. 

The new protocol always outputs $1$ for every disjoint pair $(A,B)$ because of \eqref{eq:embedding} (remember that our protocols have nonnegative outputs), and outputs $0$ most of the times for non-disjoint pairs $(A,B)$. More precisely, the probability of outputting $0$ in case $(A,B)$ is non-disjoint is at least $1 - (1-p)^\frac{1}{p} \geqslant 1 -\mathrm{e}^{-1} > 1/2$, where $\mathrm{e}$ is Euler's number. The theorem follows then directly from the fact that the new protocol has complexity $O(c/p)$ and from the fact that the set disjointness problem has randomized communication complexity $\Omega(n)$. \qed
\end{proof}

\subsection{The perfect matching polytope} 

First, we construct an embedding of the set disjointness matrix in the slack matrix of the perfect matching polytope. Then, we discuss implications for extensions of the perfect matching polytope.

\begin{lemma}
\label{lem:reduction1}
There exists an embedding of the set disjointness matrix on $[n]$ in the slack matrix of the perfect matching polytope for perfect matchings of $K^\ell$, where $\ell \leqslant 3n + 14$. 
\end{lemma}
\begin{proof}
Let $k \leqslant n+4$ denote the first multiple of $4$ that is strictly greater than $n$, and let $\ell := 3k + 2 \leqslant 3n + 14$. 

For two subsets $A$ and $B$ of $[n]$, we define an odd set $U := \alpha(A)$ and a perfect matching $M := \beta(B)$ as follows. 

First, we add the dummy element $n+1$ to $B$ in case $|B|$ is odd, so that both $B$ and $[k] - B$ contain an even number of elements. Note that this does not affect the intersection of $A$ and $B$ because $A$ is contained in $[n]$. Then, we let $U := \{i : i \in A\} \cup \{i+k : i \in A\} \cup \{3k+1\}$. 

Second, we define $M$ by adding matching edges to the partial matching $\{\{i,i+k\} : i \in [k] - B\} \cup \{\{i+k,i+2k\} : i \in B\} \cup \{\{3k+1,3k+2\}\}$ in such a way that each of the extra edges matches two consecutive unmatched vertices both in $\{i : i \in [k]\}$ or both in $\{i+2k : i \in [k]\}$. See Figure~\ref{fig:pmreduction} for an example.

It can be easily verified that $A$ and $B$ are disjoint if and only if the slack for $(U,M)$ is zero. Hence, the maps $\alpha : A \mapsto U$ and $\beta : B \mapsto M$ define the desired embedding of the set disjointness matrix. \qed
\end{proof}

\begin{figure}[ht]
\centering
\scalebox{.75}{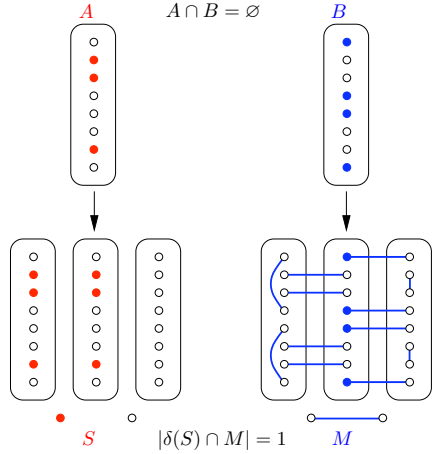}
\caption{Constructing an odd set and a perfect matching from a set disjointness instance.\label{fig:pmreduction}}
\end{figure}

Let $P$ denote the perfect matching polytope of $K^n$. Consider a size-$r$ extension of $P$ and a corresponding complexity-$\lceil \lg r \rceil$ protocol computing $S(P)$ in expectation (the existence of such a protocol is guaranteed by Theorems \ref{thm:Y91} and \ref{thm:p_vs_f}). Lemma \ref{lem:reduction1} and Theorem \ref{thm:conditional_lowerbound} together imply that $r = 2^{\Omega(n(1-\sigma^2))}$, where $\sigma^2$ is the normalized variance of the protocol. For instance, deterministic protocols for computing the slack matrix of the perfect matching polytope give rise to exponential size extensions ($\sigma^2 = 0$ in this case). The same holds if $\sigma^2$ is a constant with $0 < \sigma^2 < 1$. When $\sigma^2$ is about $(n-1)/n$ or more, the bound given by Theorem~\ref{thm:conditional_lowerbound} becomes trivial.

\subsection{Spanning tree polytopes}
\label{sec:spanning_trees}

We prove that similar results hold for the spanning tree polytope of $K^n$ as well. This is surprising, because for this polytope an extension of size $O(n^3)$ exists. 

\begin{lemma}
\label{lem:reduction2}
There exists an embedding of the set disjointness matrix on $[n]$ in the slack matrix of the spanning tree polytope of $K^{2n+1}$.
\end{lemma}
\begin{proof}
Let $\ell := 2n+1$. Recall that the rows and columns of (the non-trivial part of) the slack matrix of the spanning tree polytope of $K^\ell$ respectively correspond to subsets $U$ and spanning trees $T$. The entry for a pair $(U,T)$ is zero iff the subgraph of $T$ induced by $U$ is connected.

Given an instance of the set disjointness problem with sets $A,B\subseteq [n]$, we define $U := \alpha(A)$ and $T := \beta(B)$ as follows. For every $i \in [n]$ add the edge $\{i,2n+1\}$ to $T$. For every $i \in B$ add the edge $\{n+i,i\}$ to $T$ and for every $i \in [n] - B$ add the edge $\{n+i,2n+1\}$ to $T$. See Figure~\ref{fig:streduction} for an example.

\begin{figure}[ht]
\centering
\includegraphics[width=0.35\textwidth]{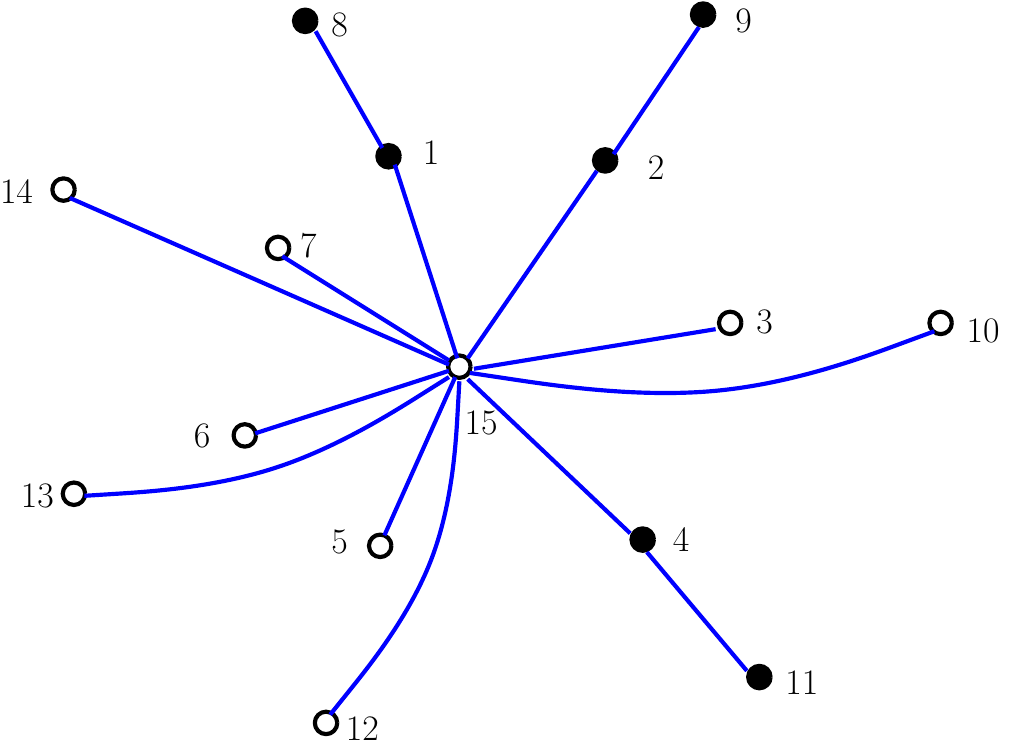}
\caption{The spanning tree $T$ for $B = \{1,2,4\}$ and $n = 7$. Black vertices are those of the form $i$ or $n+i$ where $i \in B$.\label{fig:streduction}}
\end{figure}

Finally, we let $U := \{n+i : i \in A\} \cup \{2n+1\}$. As is easily seen, $T[U]$ is connected iff $A \cap B = \varnothing$. Indeed, if $i \in A \cap B$ then $n+i$ and $2n+1$ are in different connected components of $T[U]$. Moreover, if $A \cap B = \varnothing$ then $T[U]$ is a star with $2n+1$ as center. \qed
\end{proof}

Therefore, the ``low variance forces large size'' phenomenon we exhibited for the perfect matching polytope also holds for the spanning tree polytope. Incidentally, the $O(n^3)$-size extension for the spanning tree polytope of $K^n$ can be obtained via randomized protocols, but not via deterministic ones. This is because Lemma~\ref{lem:reduction2} and Theorem~\ref{thm:conditional_lowerbound} implies that any extension for the spanning tree polytope that corresponds to a deterministic protocol must have exponential size. (Notice that the value of $p = p(n)$ for the protocol given in Subsection~\ref{subsec:MST} is roughly $1/n$.)

\section{Concluding remarks}

Given a perfect matching $M$ and an odd set $U$ as above there is always an edge in $\delta(U) \cap M$. But it is not clear if such an edge can be found using a protocol with sublinear communication. Now we show that if such an edge can be found using few bits then the perfect matching polytope has an extension of small size.

\begin{theorem}
Suppose Alice is given an odd set $U \subseteq [n]$ and Bob is given a perfect matching $M$ of $K^n$. Furthermore, suppose that Bob knows an edge $e \in \delta(U)\cap M$. Then, there exists a randomized protocol of complexity $2 \lg n + O(1)$ that computes the slack for the pair $(U,M)$ in expectation.
\end{theorem}
\begin{proof}
The protocol works as follows. Bob picks an edge $e'$ from $M \setminus\{e\}$ uniformly at random and sends it to Alice. She outputs $|M|-1=n/2-1$ if $e'\in \delta(U)$ and $0$ otherwise. The expected value of the protocol is $(|M|-1) \cdot (|\delta(U) \cap M|-1)/(|M|-1) = |\delta(U) \cap M| - 1$, as required. Bob needs to send the endpoints of the edge $e'$ to Alice and this requires $2 \lg n + O(1)$ bits. \qed
\end{proof}

The theorem above implies that if an edge in $\delta(U) \cap M$ can be computed using a protocol requiring $o(n)$ bits, then there exists an extension for the perfect matching polytope of subexponential size. We leave it as an open question to settle the existence of such a protocol.

\section*{Acknowledgements}

The authors thank Sebastian Pokutta and Ronald de Wolf for their useful feedback. The research of Faenza was supported by the German Research Foundation (DFG) within the Priority Programme 1307 Algorithm Engineering. The research of Grappe was supported by the Progetto di Eccellenza 2008--2009 of the Fondazione Cassa di Risparmio di Padova e Rovigo. The research of Fiorini was partially supported by the \emph{Actions de Recherche Concert\'ees} (ARC) fund of the French community of Belgium. The research of Tiwary was supported by the \emph{Fonds National de la Recherche Scientifique} (F.R.S.--FNRS).

\bibliographystyle{splncs}

\end{document}